\documentclass[letterpaper, 10 pt, conference]{ieeeconf}  % Comment this line out if you need a4paper

\IEEEoverridecommandlockouts                              % This command is only needed if 
\usepackage{graphics} % for pdf, bitmapped graphics files
\usepackage{epsfig} % for postscript graphics files
\usepackage{mathptmx} % assumes new font selection scheme installed
\usepackage{times} % assumes new font selection scheme installed
\usepackage{amsmath} % assumes amsmath package installed
\usepackage{amssymb}  % assumes amsmath package installed
\usepackage{float}
\usepackage{mathtools}
  \usepackage{amsthm}
\usepackage[cal=boondoxo]{mathalfa}
\usepackage{cite}
\usepackage{enumerate}
\usepackage{color}
\usepackage{booktabs}
\usepackage[ruled,linesnumbered]{algorithm2e}
\makeatletter
\newcommand{\removelatexerror}{\let\@latex@error\@gobble}
\makeatother

\newcounter{storealgline}
\IncMargin{.5em}

\DontPrintSemicolon
\makeatletter
\renewcommand{\Indentp}[1]{%
	\advance\leftskip by #1
	\advance\skiptext by -#1
	\advance\skiprule by #1}%
\renewcommand{\Indp}{\algocf@adjustskipindent\Indentp{\algoskipindent}}
\renewcommand{\Indm}{\algocf@adjustskipindent\Indentp{-\algoskipindent}}
\makeatother

\usepackage{textcomp}
\usepackage{epstopdf}
\def\BibTeX{{\rm B\kern-.05em{\sc i\kern-.025em b}\kern-.08em
		T\kern-.1667em\lower.7ex\hbox{E}\kern-.125emX}}
\makeatletter
\providecommand{\bigsqcap}{%
	\mathop{%
		\mathpalette\@updown\bigsqcup
	}%
}
\newcommand*{\@updown}[2]{%
	\rotatebox[origin=c]{180}{$\m@th#1#2$}%
}
\makeatother

\newtheorem{thm}{Theorem}

\newtheorem{prop}[thm]{Proposition}

\newtheorem{rem}{Remark}

\usepackage{tikz}
\usetikzlibrary{shapes}
\usetikzlibrary{calc}
\usetikzlibrary{arrows.meta,calc,decorations.markings,math,arrows.meta}

\newcommand{\B}{\mathcal{B}}
\newcommand{\Hk}{\mathcal{H}}

\newcommand{\n}[1]{\mathtt{n}\left(#1\right)}
\newcommand{\lag}[1]{\mathtt{L}\left(#1\right)}
\newcommand{\cs}[1]{\mathrm{colspan}\left(#1\right)}
\newcommand{\proj}[2]{\pi_{#1}(#2)}

\usepackage{xifthen}
\newcommand{\sR}[2][]{\ifthenelse{\isempty{#1}}{\mathbb{R}^{#2}}{\mathbb{R}^{#1\times #2}}}

\newcommand{\bint}[2]{{|[#1,#2]}}

\begin{document}
\title{\LARGE \bf Behavioral Systems Theory Meets Machine Learning:\\ Control-Aware Learning of the Intrinsic Behavior from Big Data}
\author{Yitao Yan, Yu Tong, Jie Bao and Wei Wang
	\thanks{Y. Yan, Y. Tong and J. Bao are with the School of Chemical Engineering, UNSW Sydney, NSW 2052, Australia. (e-mail:  y.yan@unsw.edu.au; yu.tong6@unsw.edu.au; j.bao@unsw.edu.au).}
    \thanks{W. Wang is with the Hong Kong University of Science and Technology, Hong Kong SAR, China (e-mail: weiwcs@ust.hk).}
}

\maketitle
	\thispagestyle{empty}
	\pagestyle{empty}

\begin{abstract}
The abundance of process operating data in modern industries, along with the rapid advancement of learning techniques, has led to a paradigm shift towards data-centric analysis and control. However, integrating machine learning with control theory for big data-driven control of nonlinear systems remains a challenging open problem. This is because the state-based, model-centric, and causal framework of classical control theory fundamentally contradicts the trajectory-based, set-theoretic, and causality-absent rationale of big data-based learning approaches. Using the behavioral framework, we show that dynamical systems possess an intrinsic state variable that encodes the system behavior in a bijective and causality-free manner, and control design can be carried out entirely within the state space. This approach not only resolves the aforementioned conflict but also complements machine learning techniques well, leading to a neural network architecture that is capable of learning the behavior representation well-suited for control design.

\end{abstract}

\section{Introduction}\label{sec:introduction}
The availability of massive amounts of data that are collected and stored in modern industries has created new opportunities for direct big data-based analysis and control, particularly for systems whose first-principles models are difficult or costly to obtain. This paradigm has seen remarkable developments in the context of linear time-invariant (LTI) systems. In particular, the celebrated fundamental lemma developed by Willems et al. \cite{Willems:2005} formed the basis for a wide range of data-driven approaches that bypass explicit system identification. One of the most notable developments is the data-enabled predictive control (DeePC) algorithm \cite{Coulson:2019}, which has been further developed to incorporate stability and robustness certificates \cite{Berberich:2020a}. Other developments include stability design \cite{dePersis:2020,yan2025intrinsic}, optimal control \cite{daSilva:2018}, stochastic control \cite{pan2024data} and distributed control \cite{Yan:2024}.

In contrast, the development of data-driven methods for nonlinear systems remains limited. Existing approaches either rely on linear design \cite{Berberich:2022}, which can be conservative, extend to specific classes of nonlinear systems (e.g., \cite{guo2021data,Strasser:2020}) or employ clustering-based techniques that typically handle weak nonlinearities \cite{han2025big}. With the current rapid development of artificial intelligence, machine learning methods have gained attention as a promising direction for big data-based control design due to their universal approximating capability. However, many of these approaches are centered on state-space formulations, which are often challenging due to the lack of direct state measurements \cite{strasser2026overview}. Furthermore, enforcing the classic state space representation to be learned from the dataset can introduce bias through these structural constraints. Other efforts attempt to learn step-wise encoding with a dynamic linear latent variable space \cite{Tang:2025}, but this step-wise construction can limit the level of nonlinearity captured. In recent years, Koopman operator-based methods have also emerged as a popular direction, particularly when machine learning developments enable efficient learning of Koopman eigenfunctions \cite{lusch2018deep}. However, as a theory initially developed for analyzing autonomous systems, the Koopman technique becomes challenging in handling control inputs explicitly \cite{strasser2026overview}. The above developments and challenges highlight a fundamental tension that limits the integration of classical control theory with machine learning: classical control design relies on pre-specified model structures with explicit causality, focusing on how inputs determine the evolution of system states and outputs, while machine learning, being inherently data-centric, seeks to extract system characteristics directly from large datasets without imposing explicit causal assumptions or model structures. As such, enforcing predefined structures and causality contradicts the underlying rationale of learning from data, yet discarding such structures obscures the input/output relationships, making control synthesis via classical methods difficult. 

In this paper, we take the perspective of behavioral systems theory and demonstrate how it naturally complements machine learning. This synergy gives rise to a novel approach for big-data-driven learning of system behavior, enabling integrated control design and making behavior learning inherently control-aware. We first show that dynamical systems are naturally equipped with an intrinsic state variable, whose state map captures the correspondence between trajectories in the system behavior and their free components. This allows the state map to serve as a bijective mapping, bridging the system behavior to an equipotent state behavior, enabling control design to be carried out entirely in the state space without input or output variables. Rather than following the classical \emph{control law construction} rationale, the control design adheres to the \emph{sub-behavior selection} principle from behavioral systems theory, eliminating the need for causality assumptions. This approach fits naturally with big data-based learning, making it possible to devise a {\bf C}ontrol-{\bf A}ware {\bf L}earning of {\bf I}ntrinsic {\bf B}ehavior (CALIB) architecture, which calibrates the intrinsic state such that it encapsulates system behavior while also enabling control design.

The rest of this paper is organized as follows: Section \ref{sec:preliminaries} introduces essential concepts in behavioral systems theory; Section \ref{sec:equiState} introduces the intrinsic state and illustrates its properties; Section \ref{sec:controlDesign} develops the control design approach based on the intrinsic state; Section \ref{sec:CALIB} introduces the CALIB architecture; Section \ref{sec:example} illustrates the results in this paper using an example; Section \ref{sec:conclusion} concludes the paper and points to a number of future research directions.

{\bf Notations.} We use the standard notations $\sR{}$, $\sR{n}$, $\sR[m]{n}$, $\mathbb{Z}$, etc. For a space $\mathbb{W}$, the generic variable is denoted as $w$, and its dimension is denoted by $\mathsf{w}$. For a vector $v$, $\|v\|$ and $\|v\|_\infty$ denote its two norm and infinity norm, respectively.

\section{Preliminaries}\label{sec:preliminaries}
We provide a brief summary of the concepts in behavioral systems theory that are essential to this paper. The reader is referred to \cite{Willems:1991} for more details. In this paper, we consider discrete-time dynamical systems that are time-invariant and complete. Such a dynamical system can be defined as a triple $\Sigma=(\mathbb{T},\mathbb{W},\B)$, where $\mathbb{T}\subset\mathbb{Z}$ is the time axis, $\mathbb{W}$ is the signal space, and $\B\subset\mathbb{W}^\mathbb{T}$ is the behavior that satisfies time-invariance: $\sigma\B\subset\B$ ($\sigma$ is a shifting operator such that $\sigma w_k=w_{k+1}$), and completeness: $\{w\in\B\}\Leftrightarrow\{w_\bint{k_1}{k_2}\in\B_\bint{k_1}{k_2}, \forall k_1,k_2\in\mathbb{T}, k_1\leq k_2\}$, where $\B_\bint{k_1}{k_2}$ denotes the behavior restricted to the interval $[k_1,k_2]$, i.e.,
\begin{equation}
    \B_\bint{k_1}{k_2}=\{w\mid\exists w'\in\B, w_k=w'_k, \forall k\in[k_1,k_2]\cap\mathbb{Z}\},
\end{equation}
and $w_\bint{k_1}{k_2}$ is similarly defined. In this paper, we focus on dynamical systems whose behavior has a manifold structure, i.e., those whose signal space $\mathbb{W}$ is a manifold, and behavior $\B$ is a submanifold of $\mathbb{W}^\mathbb{T}$. 

For time-invariant complete systems, there exists an integer $l$ such that $w_\bint{k}{k+l}\in\B_l, \forall k\in\mathbb{T}$ implies that $w\in\B$, and the smallest integer for which this is true is the \emph{lag} of the system denoted by $\lag{\B}$. Such a dynamical system always admits a kernel representation of the form
\begin{equation}
    r(w_\bint{k-L}{k})=0.
\end{equation}
The time-invariance and completeness allow the dynamical feature of $\B$ to be fully captured by the behavior restricted to an interval $[0,L]$ with $L\geq \lag{\B}$, which we denote as $\B_L$. Trajectories of $\B$ can be constructed from those in $\B_L$ via trajectory weaving, i.e., \cite{Markovsky:2005}
\begin{equation}
\begin{split}
    &\{w^1,w^2\in\B_L, w^1_\bint{L-l}{L}=w^2_\bint{0}{l}, l\geq\lag{\B}\}\\
    &\hskip 3cm \Rightarrow\{w^1_\bint{0}{L-l-1}\wedge w^2_\bint{0}{L}\in\B_{2L-l}\},
\end{split}
\end{equation}
where $\wedge$ denotes trajectory concatenation. The manifest variable of the system $w$ admits an input/output partition $w=(u,y)$, with $u$ as the input and $y$ as the output, if $u$ is a \emph{free variable} of the system, i.e., for all trajectories of $u$, there exists a corresponding trajectory of $y$ such that $(u,y)\in\B$. Furthermore, for a trajectory in a time-invariant complete system with a given initial trajectory $w_\bint{0}{L}$ in which $L\geq\lag{\B}$, the future output trajectory can be uniquely identified from its corresponding input trajectory.

In many cases, the description of dynamical systems relies on the use of \emph{latent variable} $\ell$, leading to a latent variable dynamical system $\Sigma^{full}=(\mathbb{T},\mathbb{W},\mathbb{L},\B^{full})$, where $\B^{full}\subset(\mathbb{W}\times\mathbb{L})^\mathbb{T}$ is the full behavior. The corresponding manifest behavior $\B$ can be constructed from $\B^{full}$ as
\begin{equation}
    \B=\pi_w(\B^{full}),
\end{equation}
where $\pi_w$ denoted the projection onto the trajectory space of $w$. The latent variable is additionally a state variable if it possesses the \emph{state property}
\begin{equation}
\begin{split}
    &\{(w^1,\ell^1), (w^2,\ell^2)\in\B^{full}, \ell_k^1=\ell_k^2\}\\
    & \hskip 0.4cm \Rightarrow \{(w^1_\bint{-\infty}{k-1}\wedge w^2_\bint{k}{\infty},\ell^1_\bint{-\infty}{k-1}\wedge\ell^2_\bint{k}{\infty})\in\B^{full}\}.
\end{split}
\end{equation}
If the latent variable has the state property, the representation of full behavior is, in general, of the form
\begin{equation}\label{eq:ssRep}
    f(x_k, \sigma x_k, w_k)=0,
\end{equation}
where $w$ is the manifest variable and $x$ is the state variable. The smallest dimension of $x$ among all such representations with manifest behavior given by $\B$ is the \emph{state cardinality} of the system, denoted by $\n{\B}$. In such a case, the state is both controllable and observable \cite{sadegh2002minimal}. This means that the initial state and the input trajectory uniquely specify a trajectory $(w,x)\in\B^{full}$, and therefore a trajectory $w\in\B$.

A distinct feature of behavioral systems theory is that control is viewed as an interconnection between the to-be-controlled system and the controller. As such, the controlled system behavior is the set of trajectories that are admissible in both the system and the controller, and is therefore a \emph{subset} of the uncontrolled behavior. This means that control design in the behavioral framework is the search for a \emph{sub-behavior} from the to-be-controlled system that meets the control objectives. This rationale transforms the control design problem into a \emph{feasibility} problem, which comprises (i) \emph{Existence}: whether the desired controlled behavior exists in the uncontrolled system behavior, (ii) \emph{Implementability}: whether this desired controlled behavior can be implemented via the manipulated variable alone, and (iii) \emph{Liveness}: whether the live component in the controlled behavior can freely choose any trajectory from its behavior before control \cite{Willems:2002}. Since the controller needs to be able to determine $w_k$ from any trajectory $w_\bint{k-L}{k-1}$  from $\B_{L-1}$ to ensure effective trajectory weaving, the live component in the interval $[k-L,k]$ is $w_\bint{k-L}{k-1}$ \cite{yan2025behavioral}. 

\section{Bridging System Behavior to an Equipotent Intrinsic State Space}\label{sec:equiState}
As illustrated in Section \ref{sec:preliminaries}, it is sufficient to learn the behavior $\B_L$ for analysis and control design for system $\Sigma$. However, constructing $\B_L$ simply using its measured trajectories poses challenges to both system analysis and control design. To illustrate, denote
\begin{equation}\label{eq:trajWindow}
    \tilde{w}_k\coloneqq w_\bint{k-L}{k}
\end{equation}
as the trajectory for the interval $[k-L,k]$. One option is to learn the system behavior as a kernel representation, i.e.,
\begin{equation}
    r(\tilde{w}_k)=0.
\end{equation}
While this poses no structural constraints on the possible behavior it can describe, it has the risk of learning a trivial mapping (i.e., $r\equiv0$). To avoid this, structural constraints on the function $r$ must be imposed, making it a less general approach than it appears \cite{Tang:2025}. Furthermore, liveness requires that $\tilde{w}_\bint{k-L}{k-1}$ allows all trajectories from $\B_{L-1}$ rather than $\mathbb{R}^{L\mathsf{w}}$, which is a challenging offline control design condition to develop for nonlinear systems. This is because $\B_L$, as a submanifold of $\sR{(L+1)\mathsf{w}}$, is of a lower dimension than $(L+1)\mathsf{w}$; hence, the live components cannot be completely free. This means that the manifest variable alone does not provide sufficient flexibility to capture the system behavior in the most efficient way.

Alternatively, one could construct a state variable of dimension $\n{\B}$ and attempt to learn a traditional input/state/output representation
\begin{equation}
    x_{k+1}=f(x_k,u_k), \ y_k=h(x_k,u_k).
\end{equation}
While this approach leads to a representation with a lag of 1 and alleviates the difficulty of liveness design, such a forced structural constraint not only limits the possibility of learning a more generalized representation given by \eqref{eq:ssRep}, but it also inherently learns \emph{causality} (i.e., trajectory evolution) rather than \emph{correspondence} (i.e., trajectory space), the latter of which is in fact what big data represents more accurately. Much like system identification approaches, such forced input/output representation with causal structural constraints has a tendency to exhibit a larger bias error \cite{Dorfler:2022}, which is undesirable for stabilization control design because it typically relies on unbiased future predictions. 

The above discussions highlight the need to construct a latent variable representation for the system using a latent variable that encodes features of the trajectory space, captures the intrinsic complexity of the system behavior, and has a state property. We now show that such a latent variable indeed exists, and it possesses additional desirable features beyond the discussion above. 
\begin{thm}\label{thm:stateMap}
    Let $\Sigma=(\mathbb{T},\mathbb{W},\B)$ be a time-invariant complete dynamical system. There exists a latent variable dynamic system $\Sigma^{full}=(\mathbb{T},\mathbb{W},\mathbb{G},\B^{full})$ with latent signal space $\mathbb{G}=\mathbb{R}^{\mathsf{g}}$ where $\mathsf{g}=(L+1)\mathsf{u}+\n{\B}$, latent variable $g$ and full behavior represented by
    \begin{equation}\label{eq:stateMap}
        g_k=\chi(\tilde{w}_k).
    \end{equation}
    Furthermore,
    \begin{enumerate}[(i)]
        \item $\chi:\B_L\rightarrow\mathbb{R}^\mathsf{g}$ is bijective;
        \item $\B$ and $\B^g$ are equipotent, where 
        \begin{equation}\label{eq:stateBehavior}
        \B^g=\proj{g}{\B^{full}};
    \end{equation}
        \item $g$ is a state variable for $\B$.
    \end{enumerate}
\end{thm}
\begin{proof}
(i) Consider a mapping $\chi$ that maps $\tilde{w}_k$ to $g_k=\mathrm{col}(x_{k-L},\tilde{u}_k)$, where $\tilde{u}_k$ is defined similarly to $\tilde{w}_k$ in \eqref{eq:trajWindow}. For two trajectory pairs $(w^1,g^1), (w^2,g^2)\in\B^{full}$, if $\chi(\tilde{w}_k^1)=\chi(\tilde{w}_k^2)$, then $g_k^1=g_k^2$, which means that $\tilde{u}_k^1=\tilde{u}_k^2$ and $x_{k-L}^1=x_{k-L}^2$. This implies that $\tilde{y}_k^1=\tilde{y}_k^2$, and hence $\tilde{w}_k^1=\tilde{w}_k^2$, showing that $\chi$ is injective. Furthermore, any choice of $g_k$ corresponds to a choice of $x_{k-L}$ and $\tilde{u}_k$. It is not difficult to see that a corresponding trajectory $\tilde{w}_k\in\B_L$ must exist. This shows that that $\chi$ is surjective. The combination of injectivity and surjectivity of $\chi$ makes it a bijective mapping between $\B_L$ and $\mathbb{R}^\mathsf{g}$. The above discussions also show that the manifest behavior of this full behavior is indeed $\B$.

(ii) This follows directly from the bijectivity of $\chi$.

(iii) Let $(w^1,g^1), (w^2,g^2)\in\B^{full}$. If $g_k^1=g_k^2$, then, by the bijectivity of $\chi$, it follows that $\tilde{w}_k^1=\tilde{w}_k^2$. Following a similar argument to the proof of Theorem 3 in \cite{yan2025intrinsic}, it follows that
\begin{equation}
    (w,g)=(w_\bint{-\infty}{k}^1\wedge w_\bint{k+1}{\infty}^2,g_\bint{-\infty}{k}^1\wedge g_\bint{k+1}{\infty}^2)\in\B^{full}.
\end{equation}
This proved the state property of $g$, making it a state variable for $\B$.
\end{proof}
\begin{rem}
    While the proof of Theorem \ref{thm:stateMap} uses initial state $x_{k-L}$ and input trajectory $\tilde{u}_k$ as a constructed example, it is not the only valid choice. To illustrate, for any valid state $g$, one can choose an invertible function $\gamma:\mathbb{R}^\mathsf{g}\rightarrow\mathbb{R}^\mathsf{g}$ and define an equivalent state as $g'=\gamma(g)$ with the associated state map $\chi'=\chi\circ\gamma^{-1}$. This flexibility is beneficial for system behavior learning because measurement of the state $x$ is typically not available in the data driven setting, and the best chart for behavior learning is not necessarily this chosen example. 
\end{rem}
\begin{rem}
    Although $(x_{k-L},\tilde{u_k})$ cannot be used as a candidate for $g$ in practice, it does reveal the rationale for why such a $g$ exists: it implicitly encodes the information of \emph{all and only} free components in the uncontrolled  behavior. This explains why $(L+1)\mathsf{u}+\n{\B}$ is the smallest possible dimension of $g$. It also shows that the representation \eqref{eq:stateMap} \emph{always} exists, making the constructions in Theorem \ref{thm:stateMap} general and not restrictive to any class of nonlinear systems.
\end{rem}

As shown in Theorem \ref{thm:stateMap}, $\chi$ maps the system trajectories to the state space, and is therefore a \emph{state map} for the system \cite{Rapisarda:1997,Willems:1998}. Furthermore, the bijectivity of $\chi$ implies the existence of an inverse mapping of $\chi$. For the clarity of presentation, we denote this inverse mapping as $\eta$, i.e.,
\begin{equation}\label{eq:behaviorChart}
    \tilde{w}_k=\eta(g_k).
\end{equation}
In such a case, the vector $g_k$ can be viewed as a chart for $\B_L$, and $\eta$ is a \emph{parameterization map}. In other words, the latent variable $g$ has chart-state dual property. This bijective nature allows the analysis of the system behavior to be carried out \emph{completely} in the state space.

Using \eqref{eq:behaviorChart} and following a similar construction in \cite{yan2025intrinsic}, the full behavior can be brought to the following equivalent representation
\begin{subequations}\label{eq:equivRep}
    \begin{align}
        \Pi_-\eta(g_k)&=\Pi_+\eta(g_{k-1}) \label{eq:weaving}\\
        w_k&=\Pi_0\eta(g_k)\label{eq:parameterization}
    \end{align}
\end{subequations}
where $\Pi_-=\begin{bmatrix} I_{L\mathsf{w}} & 0_{L\mathsf{w}\times\mathsf{w}}\end{bmatrix}$, $\Pi_+=\begin{bmatrix} 0_{L\mathsf{w}\times\mathsf{w}} & I_{L\mathsf{w}} \end{bmatrix}$ and $\Pi_0=\begin{bmatrix}
    0_{\mathsf{w}\times L\mathsf{w}} & I_\mathsf{w}
\end{bmatrix}$. This representation includes a first-order state transition that encodes manifest variable trajectory weaving and a manifest behavior parameterization using $g$ as a chart. We now show that the state behavior $\B^g$ is precisely represented by \eqref{eq:weaving}.
\begin{prop}\label{prop:stateBehavior}
    Let $\Sigma^{full}$ be a latent variable dynamical system with full behavior given by \eqref{eq:stateMap}, with the inverse mapping of $\chi$ as $\eta$, then \eqref{eq:weaving} is a representation of $\B^g$.
\end{prop}
\begin{proof}
    Since \eqref{eq:equivRep} is a representation of $\B^{full}$, the behavior $\B^g$ can be represented by \eqref{eq:weaving} because \eqref{eq:parameterization} poses no restrictions on $g$. 
\end{proof}

The combination of the results in Theorem \ref{thm:stateMap} and Proposition \ref{prop:stateBehavior} shows that the dynamics of $w$ can be analyzed and controlled by treating the equipotent state behavior $\B^g$ with a representation completely independent of $w$ and whose variable dimension captures the intrinsic complexity of $\B_L$. We therefore refer to $g$ as the \emph{intrinsic state} of $\B$.
\begin{rem}
    If $\Sigma$ is an LTI system, i.e., $\Sigma$ satisfies $\mathbb{W}$ being a vector space and $\B$ being a linear subspace of $\mathbb{W}^\mathbb{T}$, then the parameterization map can be readily obtained directly from the fundamental lemma. Specifically, let $\Hk$ denote a Hankel matrix of order $L+1$ constructed from one of its $(T+1)$-step trajectories $\tilde{w}\in\B_T$. If $\mathrm{rank}(\Hk)=(L+1)\mathsf{u}+\n{\B}$, then $\B_L$ is equal to the column span of $\Hk$ \cite{Willems:2005}. If $\Hk$ has exactly $(L+1)\mathsf{u}+\n{\B}$ columns (e.g., processed via SVD), then a bijective map can be obtained as \cite{yan2025intrinsic}
    \begin{equation}\label{eq:bijectLTI}
        g_k=\chi(\tilde{w}_k)=\Hk^\dagger \tilde{w}_k, \ \tilde{w}_k=\eta(g_k)=\Hk g_k,
    \end{equation}
    with the state behavior $\B^g$ represented as
    \begin{equation}\label{eq:weavingLTI}
        \Pi_-\Hk g_k=\Pi_+\Hk g_{k-1}.
    \end{equation}
\end{rem}

\section{Stabilization Design Using Intrinsic State: From Linear to Nonlinear Systems}\label{sec:controlDesign}
For control design in this paper, we focus on exponential stabilization of the origin. As illustrated in Section \ref{sec:equiState}, the intrinsic state allows the behavior $\B$ to be equivalently analyzed in the behavior $\B^g$. Furthermore, the state property of $g$ allows for the possibility of constructing control Lyapunov functions for stabilization design. However, since $g$ captures behavior space, it inherently has no presumed causality. Leveraging the bijective mapping \eqref{eq:stateMap} and \eqref{eq:behaviorChart}, and using the behavioral framework, which is inherently non-causal, we develop a set of feasibility conditions that construct the Lyapunov function as well as the associated \emph{controlled behavior}, as opposed to the control law (e.g., approaches in \cite{dePersis:2020}). The control action can be readily obtained from the state $g$ via the parameterization map $\eta$. 

% , and its nonlinear behavior makes it difficult to construct a virtual manipulated variable similar to the construction in \cite{yan2025intrinsic}. 

We begin this section with rationale inspired by the stabilization design of LTI systems. Since the ultimate goal is to generalize to nonlinear systems, conditions for LTI systems are developed so that the representation \eqref{eq:weavingLTI} remains unprocessed.

% , i.e., we do not compute the virtual manipulated variable from it following the approach in \cite{yan2025intrinsic}.

\subsection{An Inspiration from Data-Driven Stabilization Design for LTI Systems}
The main idea for this design is to find a controlled behavior that is parameterized by the free component of the controlled system. In this case, since $g_{k-1}$ in \eqref{eq:weavingLTI} captures the dynamics of $\tilde{w}_{k-1}$, it is the live component for the controlled behavior, and is additionally free according to the discussion in Section \ref{sec:equiState}. Motivated by this observation, we focus on finding a controlled behavior that is also LTI, i.e., the state transition in the controlled system satisfies
\begin{equation}\label{eq:controlTransition}
    g_k=\Psi g_{k-1},
\end{equation}
where $\Psi\in\sR[\mathsf{g}]{\mathsf{g}}$ is to be determined. Using the parameterization map in \eqref{eq:bijectLTI}, this means that the controlled behavior satisfies
\begin{equation}\label{eq:controlPara}
    \tilde{w}_k=\Hk g_k=\Hk\Psi g_{k-1},
\end{equation}
which shows that this controlled behavior is parameterized by $g_{k-1}$, whose freeness implies liveness, and that the controlled behavior is indeed a subset of the uncontrolled behavior because $\cs{\Hk\Psi}\subset\cs{\Hk}$. We are now in a position to develop feasibility conditions for LTI systems.

\begin{thm}\label{thm:stabilizationLTI}
    Let $\Sigma$ be an LTI system with a latent variable dynamical system whose full behavior is given by \eqref{eq:bijectLTI}. There exists a controlled LTI behavior that is exponentially stable with convergence rate $\beta\in (0,1]$ if and only if there exists a symmetric matrix $W$ and a matrix $L$ such that
    \begin{subequations}
        \begin{align}
            W&>0,\\
            \Pi_-\Hk L-\Pi_+\Hk W&=0, \label{eq:weavingLMI}\\
            \begin{bmatrix}
                (1-\beta)W & *\\ L & W
            \end{bmatrix}&\geq0.\label{eq:LyapunovLMI}
        \end{align}
    \end{subequations}
    In such a case, a possible implementing control action is
    \begin{equation}
        u_k=\Pi_u\Hk LW^{-1} g_{k-1}=\Pi_u\Hk LW^{-1}\Hk^\dagger\tilde{w}_{k-1}.
    \end{equation}
\end{thm}
\begin{proof}
    As illustrated in \cite[Theorem 8]{yan2025intrinsic}, the LTI system can be asymptotically stabilized with the controlled behavior in the form \eqref{eq:controlPara} if and only if there exists a Lyapunov function of the form $V(g) = g^\top M g$ with $M>0$ and a controlled state behavior of the form \eqref{eq:controlTransition} such that $V(g_k) - V(g_{k-1})<0$ for all $g_{k-1}\neq 0$. The exponential stability condition can therefore be analogously established as the existence of such a Lyapunov function such that
    \begin{equation}
        g_k^\top M g_k-(1-\beta)g_{k-1}^\top M g_{k-1}\leq 0,
    \end{equation}
    for all $g_{k-1}$. Since the controlled state dynamics is of the form \eqref{eq:controlTransition}, we have
    \begin{equation}
        g_{k-1}^\top((1-\beta)M - \Psi^\top M \Psi)g_{k-1}\geq0
    \end{equation}
    for all $g_{k-1}$. This means that
    \begin{equation}
        (1-\beta)M - \Psi^\top M \Psi\geq0.
    \end{equation}
    Taking Schur complement, defining $W=M^{-1}$ and $L=\Psi W$ gives \eqref{eq:LyapunovLMI}. Furthermore, to ensure the existence of the controlled behavior \eqref{eq:controlTransition}, it must be a subset of the uncontrolled behavior, which is represented by \eqref{eq:weavingLTI}. This means that all $g_k$ and $g_{k-1}$ satisfying \eqref{eq:controlTransition} must also satisfy \eqref{eq:weavingLTI}, i.e.,
    \begin{equation}
        \Pi_-\Hk \Psi g_{k-1}=\Pi_+\Hk g_{k-1}
    \end{equation}
    must hold for all $g_{k-1}$. This means that 
    \begin{equation}
        \Pi_-\Hk \Psi=\Pi_+\Hk.
    \end{equation}
    Right-multiplying both sides by $W$ and rearranging lead to \eqref{eq:weavingLMI}.

Upon satisfaction of \eqref{thm:stabilizationLTI}, the controlled behavior can be recovered as
\begin{equation}
    g_k=\Psi g_{k-1}=LW^{-1} g_{k-1},
\end{equation}
which then leads to the control action $u_k$ using $\Hk$.
\end{proof}

As shown in the proof of Theorem \ref{thm:stabilizationLTI}, the key inspiration is to determine the controlled behavior, which is a subset of the uncontrolled behavior, thereby validating the existence of the controlled behavior. The bijective mapping ensures that the stabilization of $g$ would lead to that of $w$ because linear transformations preserve the origin. Furthermore, since $g$ encodes the manipulated variable information, the controlled behavior is automatically implementable by the manipulated variable. Combined with the liveness analysis discussed in the beginning of this section, feasibility of the controlled behavior in $\B_L$, and therefore in $\B$, is guaranteed. This means that Theorem \ref{thm:stabilizationLTI} allows one to carry out control design without the need to pre-specify the controller structure, and the design in the state behavior requires no presumption of causality. The results in Theorem \ref{thm:stabilizationLTI} is therefore readily generalizable to nonlinear systems. 

\subsection{Generalization to Nonlinear Systems}
To generalize the above to a nonlinear system, we begin by assuming that the state map is anchored at 0, i.e., that
\begin{equation}
    \chi(0_{(L+1)\mathsf{w}\times 1})=0_{\mathsf{g}\times 1}.
\end{equation}
This assumption is without loss of generality because, due to the flexibility in $g$, one could always recenter the coordinates such that the mapping passes through the origin. The feasibility conditions are now immediately generalizable to nonlinear systems, which we state in the theorem below without proof.
\begin{thm}\label{thm:stabilizationNL}
    Let $\Sigma$ be a time-invariant complete system with a latent variable dynamical system whose full behavior is given by \eqref{eq:equivRep}. There exists a controlled behavior that is exponentially stable at rate $\beta\in(0,1]$ if there exist functions $V:\mathbb{R}^\mathsf{g}\rightarrow\mathbb{R}$ and $\psi:\mathbb{R}^\mathsf{g}\rightarrow \mathbb{R}^\mathsf{g}$ such that
    \begin{subequations}
        \begin{align}
            &\alpha_1(\|g_{k-1}\|)\leq V(g_{k-1})\leq \alpha_2(\|g_{k-1}\|),\label{eq:LFBound}\\
            &\Pi_-\eta\circ\psi(g_{k-1})=\Pi_+\eta(g_{k-1}),\\
            &V\circ\psi(g_{k-1})-(1-\beta)V(g_{k-1})\leq 0,
        \end{align}
    \end{subequations}
    hold for all $g_{k-1}\in\mathbb{R}^\mathsf{g}$, where $\alpha_1$ and $\alpha_2$ are two class-$\mathcal{K}_\infty$ functions. In such a case, a possible implementing control action is
    \begin{equation}
    u_k=\Pi_u\eta\circ\psi\circ\chi(\tilde{w}_{k-1}).
\end{equation}
\end{thm}

Comparing Theorem \ref{thm:stabilizationNL} with Theorem \ref{thm:stabilizationLTI}, the controlled behavior has been generalized to an arbitrary function $\psi$, making the conditions in Theorem \ref{thm:stabilizationNL} less restrictive. However, the rationale is exactly the same, i.e., the existence of $\psi$ means that the controlled behavior becomes
\begin{equation}
    \tilde{w}_k=\eta(g_k)=\eta\circ\psi(g_{k-1}),
\end{equation}
which is a subset of the uncontrolled system behavior because
\begin{equation}
    \mathrm{im}(\eta\circ\psi)\subset\mathrm{im}(\eta),
\end{equation}
where $\mathrm{im}(\eta)$ denotes the image of function $\eta$. While Theorem \ref{thm:stabilizationNL} provides conditions for the feasibility of a desired controlled behavior, determining the form of the controlled state behavior represented by $\psi$ is not straightforward, and the search for a control Lyapunov function $V$ for nonlinear systems is well-known to be a challenging task. These difficulties become more pronounced in the big data setting because the system behavior is learned from data, and the learned representation may not be suitable for control design. In the next section, leveraging the flexibility of the intrinsic state variable and the powerful ability of neural networks to capture nonlinear dynamics, we will develop a neural network architecture that learns the intrinsic state map that not only captures the system dynamics, but is also suitable for control design.

\section{A Neural Network Architecture for Intrinsic State Learning with Integrated Control}\label{sec:CALIB}
The state map $\chi$, parameterization map $\eta$, together with the conditions in Theorem \ref{thm:stabilizationNL}, form the basis for the development of a learning architecture that is suitable for nonlinear control. In this section, we introduce  the {\bf C}ontrol-{\bf A}ware {\bf L}earning of {\bf I}ntrinsic {\bf B}ehavior (CALIB) architecture. It trains the maps $\chi$ and $\eta$ together with the controlled state behavior representation $\psi$ and the control Lyapunov function $V$ simultaneously in order to learn the intrinsic state best-suited for control design. 

\begin{figure}
    \centering
    \begin{tikzpicture}[scale = 0.6]

\node (0) at (-0.5,0) {$\tilde{w}_{k-1}$};

\draw [thick] (1, 0) -- (1,-1.5) -- (2, -1) -- (2, 1) -- (1,1.5) -- (1,0);
\node at (1.5,0) {\large $\chi$};

\node (1) at (3.2,0) {$g_{k-1}$};

\draw [thick] (8, 0) -- (8,-1) -- (9, -1.5) -- (9, 1.5) -- (8,1) -- (8,0);
\node at (8.5,0) {\large $\eta$};
\node (2) at (10.5,0) {$\hat{\tilde{w}}_{k-1}$};

\draw [thick] (8, 3) -- (8,4) -- (9, 3.5) -- (9, 2.5) -- (8,2) -- (8,3);
\node at (8.5,3) {\large $V$};
\node (3) at (10.5,3) {$V_{k-1}$};

\draw [thick] (8, -3) -- (8,-4) -- (9, -3.5) -- (9, -2.5) -- (8,-2) -- (8,-3);
\node at (8.5,-3) {\large $V$};
\node (4) at (7,-3) {$g_{k}^c$};
\draw [thick] (5, -3) -- (5,-4) -- (6, -4) -- (6, -2) -- (5,-2) -- (5,-3);
\node at (5.5,-3) {\large $\psi$};
\node (5) at (10.5,-3) {$V_k$};

\draw [thick, -{Latex[scale=0.6]}] (0) -- (1,0);
\draw [thick, -{Latex[scale=0.6]}] (2,0) -- (1);

\draw [thick, -{Latex[scale=0.6]}] (1) -- (8,0);
\draw [thick, -{Latex[scale=0.6]}] (9,0) -- (2);

\draw [thick, -{Latex[scale=0.6]}] (1) -- (8,3);
\draw [thick, -{Latex[scale=0.6]}] (9,3) -- (3);

\draw [thick, -{Latex[scale=0.6]}] (1) -- (5,-3);
\draw [thick, -{Latex[scale=0.6]}] (6,-3) -- (4);
\draw [thick, -{Latex[scale=0.6]}] (4) -- (8,-3);
\draw [thick, -{Latex[scale=0.6]}] (9,-3) -- (5);

\node (6) at (-0.5,-6) {$\tilde{w}_{k}$};
\node at (1.5,-6) {\large $\chi$};

\draw [thick] (1, -6) -- (1,-7.5) -- (2, -7) -- (2, -5) -- (1,-4.5) -- (1,-6);

\node (7) at (3.2,-6) {$g_{k}$};
\node at (1.5,0) {\large $\chi$};

\draw [thick] (8, -6) -- (8,-7) -- (9, -7.5) -- (9, -4.5) -- (8,-5) -- (8,-6);
\node at (8.5,-6) {\large $\eta$};
\node (8) at (10.5,-6) {$\hat{\tilde{w}}_{k}$};

\draw [thick, -{Latex[scale=0.6]}] (6) -- (1,-6);
\draw [thick, -{Latex[scale=0.6]}] (2,-6) -- (7);

\draw [thick, -{Latex[scale=0.6]}] (7) -- (8,-6);
\draw [thick, -{Latex[scale=0.6]}] (9,-6) -- (8);
\end{tikzpicture}
    \caption{The CALIB Architecture}
    \label{fig:CALIB}
\end{figure}
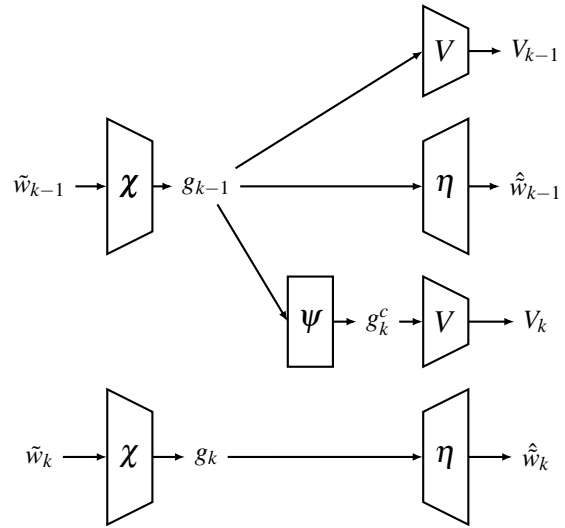

A schematic diagram of the CALIB architecture is shown in Fig. \ref{fig:CALIB}. The state map and parameterization map naturally form an autoencoder structure that reconstructs $\tilde{w}_{k-1}$, with the latent space being the intrinsic state $g_{k-1}$. A neural network representing $\psi$ is used to transition $g_{k-1}$ to the controlled state, which is denoted by $g_k^c$. Both the learned $g_{k-1}$ and the transitioned $g_k^c$ are sent to the same Lyapunov function network, producing $V_k$ and $V_{k-1}$. While learning the behavior in a single $(L+1)$-step interval is sufficient to capture the entire behavior in theory, the incomplete data set necessitates further guidance for effective learning. Specifically, to ensure that the learned state behavior reflects trajectory weaving, the trajectory $\tilde{w}_k$ whose past trajectory overlaps with $\tilde{w}_{k-1}$ is also fed into the same autoencoder. The loss function $\mathcal{L}$ is a combination of intrinsic behavior learning and controlled behavior learning, i.e.,
\begin{equation}
    \mathcal{L}=\mathcal{L}_{IB}+\mathcal{L}_{CB}.
\end{equation}
Specifically, intrinsic behavior learning loss
\begin{equation}
    \begin{split}
        \mathcal{L}_{IB}=&\|\tilde{w}_{k-1}-\eta\circ\chi(\tilde{w}_{k-1})\|_{\mathrm{MSE}}\\
        &+\lambda_\chi\|\chi(0)\|_{\mathrm{MAE}}+\lambda_\eta\|\eta(0)\|_{\mathrm{MAE}}\\
        &+\lambda_w\|\Pi_-\eta\circ\chi(\tilde{w}_{k})-\Pi_+\eta\circ\chi(\tilde{w}_{k-1})\|_{\mathrm{MSE}}\\
        &+\lambda_\infty(\left\| \tilde{w}_{k-1}-\eta\circ\chi(\tilde{w}_{k-1})\right\|_\infty+\left\|\tilde{w}_{k}-\eta\circ\chi(\tilde{w}_{k}) \right\|_\infty)
    \end{split}
\end{equation}
comprises a reconstruction error training the bijectivity of $\chi$, two anchoring losses with weightings $\lambda_\chi$ and $\lambda_\eta$ to ensure that the mapping passes through the origin, a weaving loss (with weight $\lambda_w$) to regulate the behavior of the intrinsic state, and an infinity norm loss (with weight $\lambda_\infty$) to impose an additional penalty on the worst-performing data point. The notations $\|\cdot\|_{\mathrm{MSE}}$ and $\|\cdot\|_{\mathrm{MAE}}$ denote the mean squared error and the mean absolute error, respectively. Note that $g_k$ and $g_k^c$ are not comparable as they likely correspond to two different future trajectories. The controlled behavior learning loss
\begin{equation}
\begin{split}
    \mathcal{L}_{CB}=&\lambda_e\|\Pi_-\eta\circ\psi\circ\chi(\tilde{w}_{k-1})-\Pi_+\eta\circ\chi(\tilde{w}_{k-1})\|_{\mathrm{MSE}}\\
    &+\lambda_V V\circ\chi(0)\\
    &+\lambda_\nabla\max\{V\circ\psi\circ\chi(\tilde{w}_{k-1})-(1-\beta)V\circ\chi(\tilde{w}_{k-1}),0\}
\end{split}
\end{equation}
consists of the loss to ensure the controlled behavior is a subset of the uncontrolled behavior (with weight $\lambda_e$), equilibrium point anchoring for the Lyapunov function (with weight $\lambda_V$), and exponential decay of the Lyapunov function along the controlled behavior (with weight $\lambda_\nabla$). To ensure that the Lyapunov function is bounded following \eqref{eq:LFBound}, it is preferable to construct the Lyapunov function network with additional strictures. One possibility is
\begin{equation}
    V(g)=\alpha_1(\|g\|)+(\alpha_2(\|g\|)-\alpha_1(\|g\|))S\circ\phi(g),
\end{equation}
where $\phi:\mathbb{R}^\mathsf{g}\rightarrow\mathbb{R}$ is a neural network to be trained, and $S$ denotes the sigmoid function.

The CALIB architecture is a combination of the learning of system behavior and controlled behavior, training the state map $\chi$, the parameterization map $\eta$, the controlled behavior representation $\psi$, and the control Lyapunov function $V$ simultaneously. As such, the learned intrinsic state not only captures the system complexity in a generic way but is also suitable for control design. This design highlights the interplay between behavioral systems theory and machine learning: the former guarantees the feasibility of the latter, while the latter realizes the former. The combined effort leads to an architecture that imposes almost no structural constraints on the underlying system behavior or machine learning components beyond those that are inherent to the system.

\section{Illustrative Example}\label{sec:example}
We consider a rigid-body drone in the same setting as in \cite{SAVKIN2023154}. The discretized model of the drone is given by
\begin{subequations}
    \begin{align}
x_{k+1} &= x_k + \tau v_k \cos(\mu_k), \\
y_{k+1} &= y_k + \tau v_k \sin(\mu_k), \\
\mu_{k+1} &= \mu_k + \tau \omega_k,\\
z_{k+1} &= z_k + \tau s_{k},
\end{align}
\end{subequations}
where $x, y, z$ are the coordinates of the spatial location of the drone in meters, $\mu$ is the yaw in radians, and $\tau$ is the sampling period, which is chosen as $0.1 \mathrm{s}$ for this example. Manipulated variables of the drone include the forward velocity $v$ in $\mathrm{m/s}$, yaw rate $\omega$ in $\mathrm{rad/s}$ and vertical velocity $s$ in $\mathrm{m/s}$. Note that the model is only used for data generation purposes and is not used during offline design or online implementation. Furthermore, the measured system output $y$ contains only the position coordinates $x, y, z$ and does not have access to the yaw angle information. The system manifest variable is therefore $w=\mathrm{col}(x,y,z,v,\omega,s)$.

In this example, we aim to stabilize the drone at the origin $(0,0,0)$, which corresponds to the operating input $v=0$, $\omega=0$, and $s=0$. We begin by collecting trajectory measurements near the origin to demonstrate the validity of the linear design, and subsequently show the control outcome of CALIB with an initial position far from the origin. 
\subsection{The Linear Design Approach}
Simulation has been performed with $L=5$ using the design method for LTI systems described in Theorem \ref{thm:stabilizationLTI} using a Hankel matrix constructed from a measured trajectory, and the control result is shown in Fig. \ref{fig:ys}. The controlled system has been stabilized to a position that is very close to the origin, and this small offset is because the behavior representation is constructed using data from the nonlinear system. The control performance is compared against the DeePC algorithm, with results shown in Fig. \ref{fig:ysds}. Compared with DeePC, the approach in Theorem \ref{thm:stabilizationLTI} leads to a much smoother trajectory that ends in a position much closer to the setpoint.
\begin{figure}
    \centering
    \includegraphics[width=0.8\linewidth]{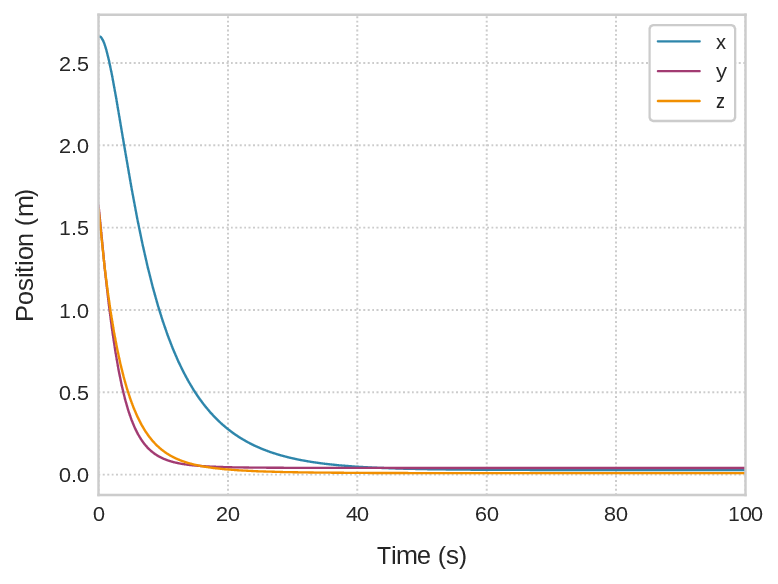}
    \caption{Control performance using the linear design in Theorem \ref{thm:stabilizationLTI}}
    \label{fig:ys}
\end{figure}

\begin{figure}
    \centering
    \includegraphics[width=0.8\linewidth]{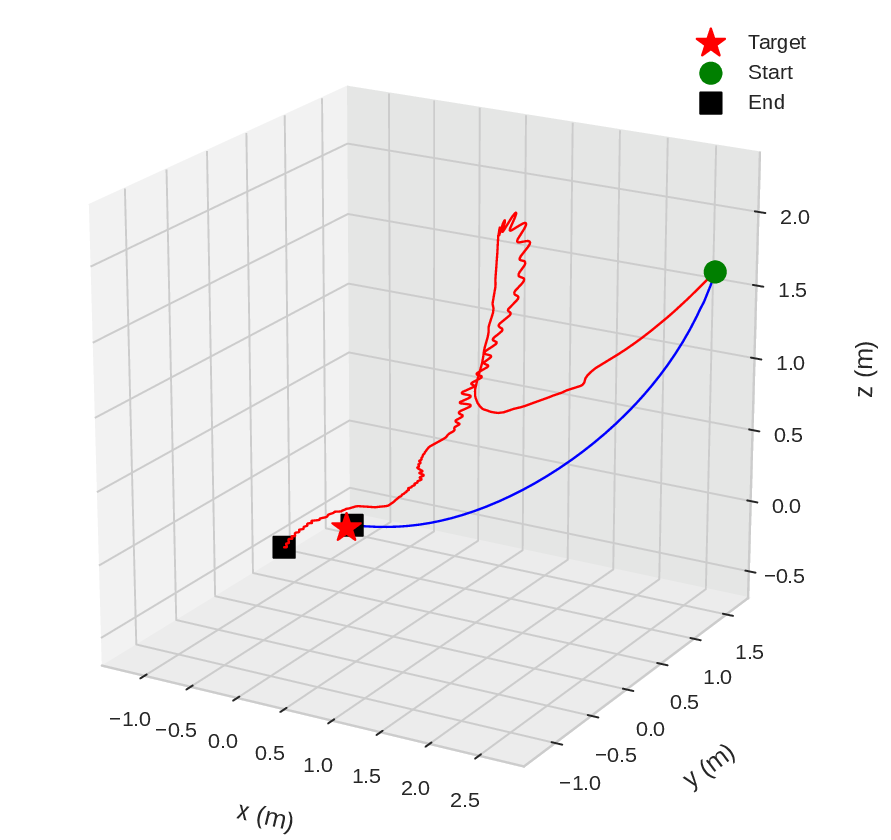}
    \caption{Performance comparison between the design in Theorem \ref{thm:stabilizationLTI} (blue) and DeePC (red)}
    \label{fig:ysds}
\end{figure}
\subsection{CALIB Learning and Control Performance}
To train the CALIB network, 110,000 trajectories (100k/10k for training/testing) of length 201 have been collected. The input trajectories are random sequences from the ranges $v\in[0,5]$, $\omega\in[-0.5,0.5]$, and $u_v\in[-5,5]$ with random hold periods between 1 and 5. For this study, we choose $L=4$, which means that the dimension of $g$ is 19. The data is normalized using z-scores for training, and the hyperparameters chosen for training are $\lambda_\chi=\lambda_\eta=\lambda_w=\lambda_e=\lambda_V=1$, $\lambda_\infty=10^{-5}$ and $\lambda_\nabla=10$. This places equal importance on reconstruction, state map anchoring, trajectory weaving, and Lyapunov function anchoring, with an emphasis on Lyapunov stability along the controlled behavior.

The control performance is shown in Fig.~\ref{fig:position}, with the 3D trajectory plot shown in Fig.~\ref{fig:3d_trajectory}. They show that the drone is stabilized to a position that is close to the origin. The offset at steady state may be caused by the fact that, while the errors in the loss function can be trained to small quantities, their effects become significant when the position is close to the origin. As a comparison, the trajectory predicted by CALIB has also been included in the figures, showing that the predicted trajectory aligns well with the actual one. This example demonstrates the ability of CALIB to capture system behavior while carrying out effective control design simultaneously.

\begin{figure}
    \centering
    \includegraphics[width=0.8\linewidth]{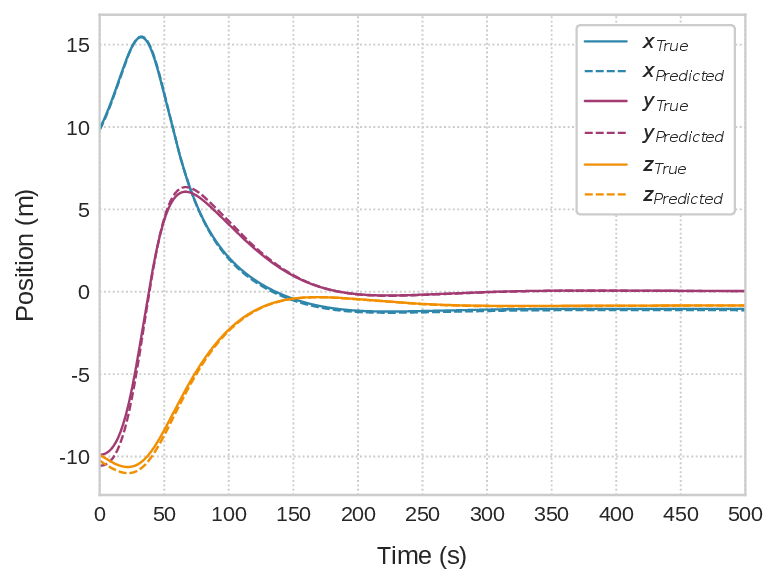}
    \caption{Control performance using the CALIB architecture}
    \label{fig:position}
\end{figure}

\begin{figure}
    \centering
    \includegraphics[width=0.8\linewidth]{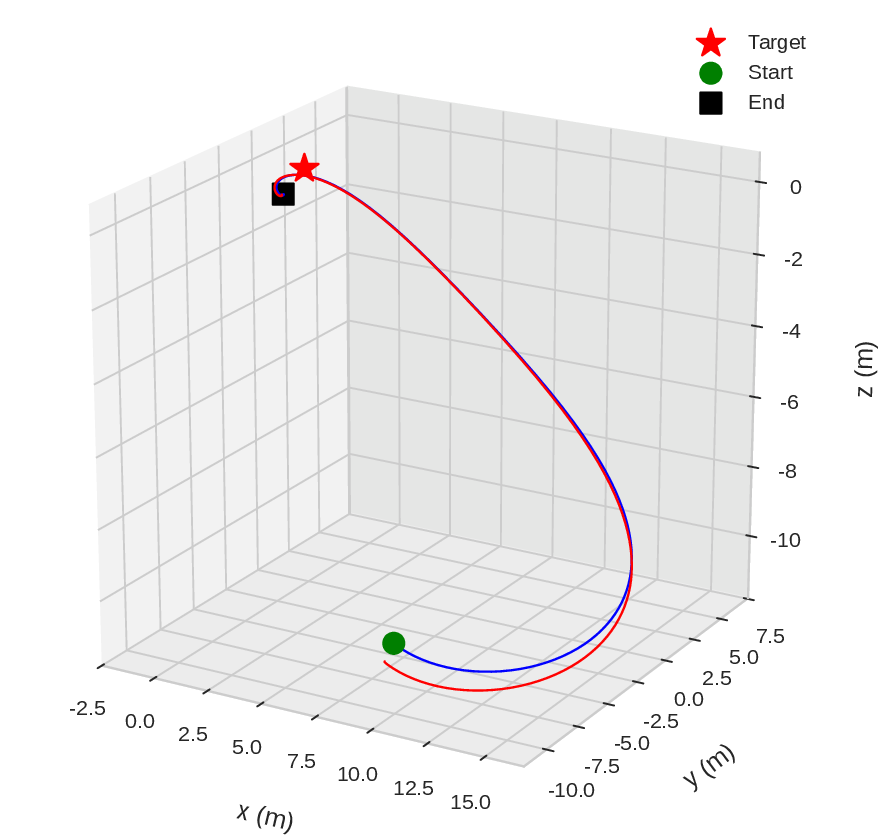}
    \caption{Comparison between true flight trajectory (blue) and that predicted using CALIB (red)}
    \label{fig:3d_trajectory}
\end{figure}

\section{Conclusion and Outlooks}\label{sec:conclusion}
    In this paper, a new approach has been developed that enables the seamless integration of machine learning and control design for big data-driven learning and control of nonlinear systems. Central to this approach is the introduction of the concept of the intrinsic state variable, which captures the intrinsic complexity of the system behavior, allowing the system behavior to be represented through an equipotent state behavior. The system behavior can be characterized by an equipotential state behavior. By developing a control design approach using this intrinsic state in the behavioral framework, we have developed the CALIB architecture that learns the system behavior integrated with control design. We believe that these developments offer a fresh perspective that is useful in big data-based nonlinear control. 
    
    This approach also opens up several avenues for future research. Firstly, the current development assumes that the data are noise-free. Robustification of the framework is, therefore, necessary to handle measurement noise. Furthermore, the CALIB design can be refined to mitigate the effects of training errors on control performance. Another interesting direction is to extend the proposed approach for setpoint-independent tracking control design. This would mean that the system behavior is endowed with additional structures, which may lead to an enriched CALIB architecture.

\bibliographystyle{IEEEtran}
\bibliography{refUpdated}

\end{document}